\documentclass{article}
\usepackage{fullpage}

\usepackage{amsmath}
\usepackage{amssymb}
\usepackage{stmaryrd}
\usepackage{graphicx}
\usepackage{hyperref}
\usepackage{color}

\usepackage{amsthm}
\usepackage{proof}

\usepackage{mathtools}
\usepackage{thmtools}

\usepackage{multicol}

\usepackage{caption}
\usepackage{subcaption}

\declaretheorem[name=Theorem,numberwithin=section]{theorem}

\declaretheorem[sibling=theorem]{corollary}

\declaretheorem[sibling=theorem]{proposition}

\declaretheorem[sibling=theorem]{definition}
\declaretheorem[sibling=theorem]{example}

\newcommand\definand[1]{{\bf #1}}

\newcommand\defeq{\stackrel{\text{\tiny def}}{=}}

\title{Counting isomorphism classes of $\beta$-normal linear lambda terms}
\author{Noam Zeilberger}

\begin{document}
\maketitle

\begin{abstract}
Unanticipated connections between different fragments of lambda calculus and different families of embedded graphs (a.k.a.~``maps'') motivate the problem of enumerating $\beta$-normal linear lambda terms.
In this brief note, it is shown (by appeal to a theorem of Arqu\`es and Beraud) that the sequence counting isomorphism classes of $\beta$-normal linear lambda terms up to free exchange of adjacent lambda abstractions coincides with the sequence counting isomorphism classes of rooted maps on oriented surfaces (A000698).
\end{abstract}

\section{Introduction}
\label{sec:intro}

Recent studies of the combinatorics of linear lambda calculus have revealed some surprising connections to the theory of \emph{graphs on surfaces} \cite{landozvonkin}.
In \cite{bodini-et-al}, Bodini, Gardy, and Jacquot gave a size-preserving bijection between $\alpha$-equivalence classes of closed linear lambda terms and a certain family of embedded graphs equivalent to \emph{rooted trivalent maps}.
Recall that a \emph{trivalent map} is a cellular embedding of an undirected 3-regular graph (possibly containing loops and multiple edges) into a compact oriented surface without boundary (or ``oriented surface'' for short), and that a \emph{rooting} of a map is the choice of an edge equipped with a direction.
Equivalently, a rooted trivalent map may be defined in purely algebraic terms, as a pair of fixed-point-free permutations $(\sigma,\alpha)$ acting transitively on a set $D$ and such that $\sigma^3 = \alpha^2 = 1$, together with a chosen element $r \in D$ (cf.~\cite{jones-singerman94schneps,vidal2010}).
Under this formulation, two rooted maps are isomorphic just in case one can be obtained from the other by conjugation along a root-preserving bijection.
In the context of lambda calculus, the adjective ``linear'' means that every variable in a term (whether free or introduced by a lambda abstraction) is used exactly once.
For example, the terms
$\lambda x.\lambda y.yx$, $\lambda x.x(\lambda y.y)$, and $\lambda x.\lambda y.xy$
are linear, while the terms
$\lambda x.xx$, $\lambda x.\lambda y.x$, and $\lambda x.\lambda y.y$
are non-linear.
The most basic relationship on lambda terms is \emph{$\alpha$-conversion}: essentially, two lambda terms are $\alpha$-equivalent if one can be obtained from the other by renaming of variables (e.g., $\lambda x.\lambda y.yx \equiv \lambda a.\lambda b.ba$, but $\lambda x.\lambda y.yx \not\equiv \lambda a.\lambda b.ab$).

After $\alpha$-equivalence, the next most fundamental relationship in lambda calculus is that induced by the rule of $\beta$-reduction $(\lambda x.t)(u) \to t[u/x]$.
A term is said to be \emph{$\beta$-normal} if it contains no subterms of the form $(\lambda x.t)(u)$.
In \cite{zg2015rpmnpt}, Alain Giorgetti and I showed that ($\alpha$-equivalence classes of) closed $\beta$-normal linear lambda terms satisfying an additional property of planarity are in size-preserving bijection with \emph{rooted planar maps}.
Recall that a planar map is a cellular embedding of an undirected graph (now with arbitrary vertex degrees) into the (oriented) sphere, and that it can be rooted by choosing an edge equipped with a direction.
Equivalently, a rooted map (on an oriented surface of arbitrary genus) may be represented by a pair of permutations $(\sigma,\alpha)$ acting transitively on a set $D$ and such that $\alpha$ is a fixed-point-free involution, together with a chosen element $r \in D$, again considered up to root-preserving conjugation.\footnote{Note that in general, the special interest of \emph{rooted} maps (as opposed to unrooted maps) from the standpoint of combinatorics is that they have no symmetries (i.e., their automorphism groups are trivial). For an interesting historical account of this motivation, see the beginning of Chapter 10 of Tutte's \emph{Graph Theory As I Have Known It} (Oxford, 1998).}
Then the map is planar just in case it satisfies Euler's formula $c(\sigma) - c(\alpha) + c(\phi) = 2$, where the ``faces'' permutation is defined by $\phi \defeq \alpha^{-1};\sigma^{-1}$ and $c(\pi)$ counts the number of cycles in $\pi$.

The connection between $\beta$-normal planar lambda terms and rooted planar maps was discovered without awareness of the work in \cite{bodini-et-al} and quite by accident, by querying the Online Encyclopedia of Integer Sequences.
Moreover, I believe it is accurate to say that the relationship between these two separate bijections linking lambda calculus and maps is for the moment not well-understood.
Still, the existence of a table of correspondences
\begin{center}
\begin{tabular}{|c||c|}
\hline
linear lambda terms & rooted trivalent maps on oriented surfaces \\
\hline
$\beta$-normal planar lambda terms & rooted maps on the sphere \\
\hline
\end{tabular}
\end{center}
makes it natural to ask what happens when one looks at $\beta$-normal linear lambda terms in general, without the condition of planarity.

In this note, I want to present a strong piece of combinatorial evidence that this table can indeed be continued in the way that one might hope for, provided that we look at lambda calculus through the right lens.
Specifically, I will prove via generating functions that the sequence counting closed $\beta$-normal linear lambda terms up to a natural notion of isomorphism (namely, \emph{free exchange of adjacent lambda abstractions}) coincides with the well-known sequence (A000698) counting rooted maps on oriented surfaces (as well as various other families of objects, as described in the OEIS entry).
Moreover, this correspondence is valid in two indices, in fact establishing that the cardinality of the set of isomorphism classes of \emph{neutral} linear lambda terms with a given size $n$ and given number of free variables $k$ equals the cardinality of the set of isomorphism classes of rooted maps with $n$ edges and $k$ vertices. 
Although I will not present an explicit bijection here, the fact that there is a simple numerical correspondence suggests that it could be fruitful to look for further relationships, and to try to understand their deeper causes.

\section{Enumerating $\beta$-normal linear lambda terms up to isomorphism}

\newcommand\GFlin{L}
\newcommand\GFpla{P}
\newcommand\GFplaB{\GFpla_B}
\newcommand\GFplaR{\GFpla_R}
\newcommand\GFlinB{\GFlin_B}
\newcommand\GFlinR{\GFlin_R}
\newcommand\GFlinqB{\tilde{\GFlin}_B}
\newcommand\GFlinqR{\tilde{\GFlin}_R}

Before considering the problem, let's recall how to enumerate linear lambda terms without regard to $\beta$-reduction.
We will always be considering terms modulo $\alpha$-equivalence, so from now we will just write ``term'' as a shorthand for ``$\alpha$-equivalence class of terms''.
Letting $t_{n,k}$ stand for the number of terms with $n$ total (free or bound) variable occurrences and $k$ free variables, we can verify that the generating function
$$\GFlin(z,x) = \sum_{n,k} t_{n,k} \frac{x^k z^n}{k!}$$
satisfies the following functional equation:
\begin{equation}
\GFlin(z,x) = zx + \GFlin(z,x)^2 + \frac{\partial}{\partial x}\GFlin(z,x)
\label{GFlin}
\end{equation}
Intuitively, this equation (which is essentially the one given in \cite{bodini-et-al}, up to reindexing) expresses the fact that any linear lambda term is either a variable, an application of one term to another term, or the lambda abstraction of a term in one of its free variables.
\emph{Closed} linear lambda terms are thus enumerated by the ordinary generating function $\GFlin(z,0)$.
Using this generating function one can for instance easily calculate the first values of $t_{n,0}$ (starting at $n = 1$),
$$ 1, 5, 60, 1105, 27120, 828250, \dots $$
and verify that these coincide with the first terms of series A062980 in the OEIS, counting rooted trivalent maps on oriented surfaces by number of edges (starting at $n = 0$, i.e., the index is shifted by one).

To count terms containing no subterms of the form $(\lambda x.t)(u)$, we can use the following standard characterization of ($\beta$-)normal terms in mutual induction with so-called \emph{neutral} terms:
\begin{enumerate}
\item Any variable is neutral.
\item If $t$ is neutral and $u$ is normal then the application $t(u)$ is neutral.
\item Every neutral term is also normal.
\item If $t$ is normal and $x$ is a free variable in $t$ then the abstraction $\lambda x.t$ is normal.
\end{enumerate}
As the notion of ``size'', it turns out to be natural (see \cite{zg2015rpmnpt}) to count the total number of times that rule 3 is invoked on subterms in the \emph{proof} that a term is neutral or normal---for a normal term this equals the total number of variable occurrences, but for a neutral term it is one less than that number.
Equation (\ref{GFlin}) thus splits into the following functional equations for the generating functions $\GFlinB(z,x)$ and $\GFlinR(z,x)$ counting neutral and normal linear terms by size and number of free variables (here $B$ and $R$ stand for ``blue'' and ``red'', following the color scheme blue = neutral, red = normal):
\begin{align}
\GFlinB(z,x) &= x + \GFlinB(z,x)\GFlinR(z,x) \label{GFlinB}\\
\GFlinR(z,x) &= z\GFlinB(z,x) + \frac{\partial}{\partial x}\GFlinR(z,x) \label{GFlinR}
\end{align}
Note that the planar case is obtained by replacing the derivative in (\ref{GFlinR}) by a ``discrete derivative'':
\begin{align}
\GFplaB(z,x) &= x + \GFplaB(z,x)\GFplaR(z,x) \label{GFplaB}\\
\GFplaR(z,x) &= z\GFplaB(z,x) + \frac{1}{x}(\GFplaR(z,x)-\GFplaR(z,0)) \label{GFplaR}
\end{align}
In particular, as established in \cite{zg2015rpmnpt}, the sequence enumerated by $\GFplaR(z,0)$
$$ 1, 2, 9, 54, 378, 2916,\dots $$ 
is OEIS series A000168, counting rooted planar maps by number of edges.
Unfortunately, the sequence enumerated by $\GFlinR(z,0)$, which begins
$$ 1, 3, 26, 367, 7142, 176766,\dots $$
does \emph{not} seem to match any known sequence related to maps.
In particular, it overshoots the sequence
$$ 1, 2, 10, 74, 706, 8162,\dots $$
counting rooted maps on oriented surfaces (A000698) by a large margin.
So, we need to try to understand why there seem to be ``too many'' $\beta$-normal linear lambda terms.

Well, if we spend some time looking more closely at, say, the first 30 such terms, 
\begin{multicols}{5}
\scriptsize
\begin{enumerate}
\item $\lambda x. x$
\item $\lambda x. x(\lambda y. y)$
\item $\lambda x. \lambda y. x(y)$
\item $\lambda x. \lambda y. y(x)$
\item $\lambda x. x(\lambda y. y(\lambda z. z))$
\item $\lambda x. x(\lambda y. \lambda z. y(z))$
\item $\lambda x. x(\lambda y. \lambda z. z(y))$
\item $\lambda x. x(\lambda y. y)(\lambda z. z)$
\item $\lambda x. \lambda y. x(y)(\lambda z. z)$
\item $\lambda x. \lambda y. y(x)(\lambda z. z)$
\item $\lambda x. \lambda y. x(y(\lambda z. z))$
\item $\lambda x. \lambda y. x(\lambda z. y(z))$
\item $\lambda x. \lambda y. x(\lambda z. z(y))$
\item $\lambda x. \lambda y. x(\lambda z. z)(y)$
\item $\lambda x. \lambda y. y(x(\lambda z. z))$
\item $\lambda x. \lambda y. y(\lambda z. x(z))$
\item $\lambda x. \lambda y. y(\lambda z. z(x))$
\item $\lambda x. \lambda y. y(\lambda z. z)(x)$
\item $\lambda x. \lambda y. \lambda z. x(y)(z)$
\item $\lambda x. \lambda y. \lambda z. y(x)(z)$
\item $\lambda x. \lambda y. \lambda z. x(z)(y)$
\item $\lambda x. \lambda y. \lambda z. z(x)(y)$
\item $\lambda x. \lambda y. \lambda z. x(y(z))$
\item $\lambda x. \lambda y. \lambda z. x(z(y))$
\item $\lambda x. \lambda y. \lambda z. y(z)(x)$
\item $\lambda x. \lambda y. \lambda z. z(y)(x)$
\item $\lambda x. \lambda y. \lambda z. y(x(z))$
\item $\lambda x. \lambda y. \lambda z. y(z(x))$
\item $\lambda x. \lambda y. \lambda z. z(x(y))$
\item $\lambda x. \lambda y. \lambda z. z(y(x))$
\end{enumerate}
\end{multicols}
\noindent
we will eventually realize that many of these terms do essentially the same thing.
For example, the terms $\lambda x. \lambda y. x(y)$ and $\lambda x. \lambda y. y(x)$ both represent a function which takes a pair of arguments and applies one to the other: the difference is just in whether the \emph{first} argument is applied to the \emph{second}, or vice versa.
This explains some of the motivation for the following definition:
\begin{definition}
Let $t$ be a linear lambda term.
We say that $t'$ is a \definand{local exchange} of $t$ if it is obtained (up to $\alpha$-equivalence) by replacing some subterm $\lambda x.\lambda y.u$ in $t$ by the term $\lambda y.\lambda x.u$.
We say that two terms $t_1$ and $t_2$ are \definand{isomorphic (up to free local exchange)} $t_1 \cong t_2$ if $t_2$ can be obtained from $t_1$ by a sequence of local exchanges.
\end{definition}
\begin{example}
$\lambda x. \lambda y. x(y) \cong \lambda x. \lambda y. y(x)$, since 
$\lambda y. \lambda x. x(y)$ is a local exchange of $\lambda x. \lambda y. x(y)$, and is $\alpha$-equivalent to $\lambda x. \lambda y. y(x)$.
Likewise, $\lambda x. \lambda y. y(\lambda z.\lambda w.z(w))(x) \cong \lambda x. \lambda y. x(\lambda z.\lambda w.w(z))(y)$.
On the other hand, observe that $\lambda x.\lambda y.y(\lambda z.x(z))$ cannot be obtained from $\lambda x.\lambda y.y(\lambda z.z(x))$ by local exchanges (that would rather require a ``non-local'' exchange of $\lambda x$ with $\lambda z$).
\end{example}
\noindent
Suppose we group the first 30 closed $\beta$-normal linear lambda terms into isomorphism classes:
\begin{multicols}{3}
\scriptsize
\begin{enumerate}
\item $\lambda x. x$
\item $\lambda x. x(\lambda y. y)$
\item $\lambda x. \lambda y. x(y)$ \\ $\lambda x. \lambda y. y(x)$
\item $\lambda x. x(\lambda y. y(\lambda z. z))$
\item $\lambda x. x(\lambda y. \lambda z. y(z))$ \\ $\lambda x. x(\lambda y. \lambda z. z(y))$
\item $\lambda x. x(\lambda y. y)(\lambda z. z)$
\item $\lambda x. \lambda y. x(y)(\lambda z. z)$ \\ $\lambda x. \lambda y. y(x)(\lambda z. z)$
\item $\lambda x. \lambda y. x(y(\lambda z. z))$ \\ $\lambda x. \lambda y. y(x(\lambda z. z))$
\item $\lambda x. \lambda y. x(\lambda z. y(z))$ \\ $\lambda x. \lambda y. y(\lambda z. x(z))$
\item $\lambda x. \lambda y. x(\lambda z. z(y))$ \\ $\lambda x. \lambda y. y(\lambda z. z(x))$
\item $\lambda x. \lambda y. x(\lambda z. z)(y)$ \\ $\lambda x. \lambda y. y(\lambda z. z)(x)$
\item $\lambda x. \lambda y. \lambda z. x(y)(z)$ \\ $\lambda x. \lambda y. \lambda z. y(x)(z)$ \\ $\lambda x. \lambda y. \lambda z. x(z)(y)$ \\ $\lambda x. \lambda y. \lambda z. z(x)(y)$ \\ $\lambda x. \lambda y. \lambda z. y(z)(x)$ \\ $\lambda x. \lambda y. \lambda z. z(y)(x)$
\item $\lambda x. \lambda y. \lambda z. x(y(z))$ \\ $\lambda x. \lambda y. \lambda z. x(z(y))$ \\ $\lambda x. \lambda y. \lambda z. y(x(z))$ \\ $\lambda x. \lambda y. \lambda z. y(z(x))$ \\ $\lambda x. \lambda y. \lambda z. z(x(y))$ \\ $\lambda x. \lambda y. \lambda z. z(y(x))$
\end{enumerate}
\end{multicols}
\noindent
Counting the number of isomorphism classes at a given size $n$, we can verify that this sequence matches the first few entries of A000698: one term for $n = 1$, two isomorphism classes for $n = 2$, and ten isomorphism classes for $n = 3$.
Indeed, we can prove that this coincidence continues indefinitely.
\begin{proposition}
The generating functions $\GFlinqB(z,x)$ and $\GFlinqR(z,x)$ counting isomorphism classes of neutral and normal linear lambda terms up to free local exchange satisfy the following functional equations:
\begin{align}
\GFlinqB(z,x) &= x + \GFlinqB(z,x)\GFlinqR(z,x) \label{GFlinqB} \\
\GFlinqR(z,x) &= z\sum_{i=0}^\infty \frac{1}{i!}\cdot \frac{\partial^i}{\partial x^i}\GFlinqB(z,x) \label{GFlinqR}
\end{align}
\end{proposition}
\begin{proof}
Equation (\ref{GFlinqB}) is analogous to (\ref{GFlinB}), while (\ref{GFlinqR}) is justified as follows.
Any normal term may be constructed by picking $i$ distinct variables $x_1,\dots,x_i$ which are free in some neutral term $t$ and lambda-abstracting them to produce $\lambda x_1.\dots\lambda x_i.t$.
This construction is accounted for by the $i$th partial derivative $\frac{\partial^i}{\partial x^i}\GFlinqB(z,x)$.
But since any way of ordering the variables yields a normal term in the same isomorphism class up to local exchanges, we must divide by a factor of $i!$, and since $i$ here is arbitrary we take a sum indexed over all non-negative integers.
Finally, since the size of the resulting normal term is (by definition) one plus the size of its neutral subterm $t$, we multiply by $z$.
\end{proof}
\begin{theorem}
The coefficient of $z^n x^k$ in $\GFlinqB(z,x)$ gives the number of isomorphism classes of rooted maps with $n$ edges and $k$ vertices on oriented surfaces of arbitrary genus.
\end{theorem}
\begin{proof}
The expression $\sum_{i=0}^\infty \frac{1}{i!}\cdot \frac{\partial^i}{\partial x^i}\GFlinqB(z,x)$ appearing in (\ref{GFlinqR}) is just the Taylor expansion of $\GFlinqB(z,x+1)$, so by substitution into (\ref{GFlinqB}) we have that
\begin{align}
\GFlinqB(z,x) = x + z\GFlinqB(z,x)\GFlinqB(z,x+1). \label{ArquesBeraud}
\end{align}
But Arqu\`es and Beraud have shown that the two variable generating function counting oriented rooted maps by edges and vertices is precisely the solution to equation (\ref{ArquesBeraud}),
as Corollary 2 in \cite{arques-beraud1}.
\end{proof}
\begin{corollary}
The number of isomorphism classes of rooted maps with $n$ edges is equal to the number of isomorphism classes of closed $\beta$-normal linear lambda terms with size $n+1$.
\end{corollary}
\begin{proof}
Since $\GFlinqR(z,0) = z\GFlinqB(z,1)$.
\end{proof}
\medskip

\noindent
{\bf Acknowledgments.}
I am grateful to Maciej Do\l{}ega and Alain Giorgetti for many helpful conversations about rooted maps and their potential connections to linear lambda terms.

\bibliographystyle{abbrvnat}

\end{document}